\documentclass[11pt]{article}
\usepackage{graphicx}
\graphicspath{{./figs/}} 
\usepackage{subfig}
\usepackage{float}
\usepackage{amsmath,amssymb,amsfonts, amsthm}
\usepackage{listings}
\usepackage{color}
\usepackage{xcolor}
\usepackage{textcomp}
\newtheorem{theorem}{Theorem}[section]

\newtheorem{lemma}{Lemma}[section]
\newtheorem{definition}{Definition}[section]
\newtheorem{assumption}{Assumption}[section]
\theoremstyle{definition}
\newtheorem{remark}{Remark}[section]  
\newtheorem{example}{Example}[section]

\begin{document}

\title{A No-Arbitrage Model of Liquidity in Financial Markets involving
Brownian Sheets\thanks{With the collaboration of Thanh Hoang}}
\author{David German\footnote{Claremont McKenna College} and Henry
  Schellhorn\footnote{Claremont Graduate University}}
\date{\today }
\maketitle


\begin{abstract}
We consider a dynamic market model where buyers and sellers submit limit
orders. If at a given moment in time, the buyer is unable to complete his
entire order due to the shortage of sell orders at the required limit price,
the unmatched part of the order is recorded in the order book. Subsequently
these buy unmatched orders may be matched with new incoming sell orders. The
resulting demand curve constitutes the sole input to our model. The clearing
price is then mechanically calculated using the market clearing condition.
We use a Brownian sheet to model the demand curve, and provide some
theoretical assumptions under which such a model is justified. 

Our main result is the proof that if there exists a unique equivalent
martingale measure for the clearing price, then under some mild
assumptions there is no arbitrage. We use the Ito-Wentzell formula to
obtain that result, and also to characterize the dynamics of the
demand curve and of the clearing price in the equivalent measure. We
find that the volatility of the clearing price is (up to a stochastic factor)
inversely proportional to the sum of buy and sell order flow density
(evaluated at the clearing price), which confirms the intuition that
volatility is inversely proportional to volume. We also demonstrate
that our approach is implementable. We use real order book data and
simulate option prices under a particularly simple parameterization of
our model.

The no-arbitrage conditions we obtain are applicable to a wide class of
models, in the same way that the Heath-Jarrow-Morton conditions apply to a
wide class of interest rate models. 
\end{abstract}

\section{Introduction}

Most liquidity models in mathematical finance abstract the trading
mechanism from the characterization of prices in the resulting
market. Our viewpoint is fundamentally different. In our model the
equilibrium prices of the assets are completely determined by the
order flow, which is viewed as an exogenous process. We model a market
of assets without a specialist, where every trader submits limit
orders, that is, for a buy order, the buyer specifies the maximum
price, or the buy limit price, that he/she is willing to pay, and, for
a sell order, the seller specifies the minimum price, or the sell
limit price, at which he/she is willing to sell \footnote{ There is no
  loss of generality in this statement. A buy market order can be
  specified in our model as a buy limit order with the limit price
  equal to infinity. Since we model assets with only positive prices,
  a sell market order can be specified in our model as a sell limit
  order with a limit price equal to zero.}.

If at a given moment in time, the buyer is unable to complete his
entire order due to the shortage of sell orders at the required limit
price, the unmatched part of the order is recorded in the order
book. A symmetric outcome follows in the case of incoming sell
orders. Subsequently these buy unmatched orders may be matched with
new incoming sell orders. We note that many electronic exchanges, such
as NYSE Arca \cite{arca:11}, operate like this.  Time-priority is used
to break indeterminacies of a match between an incoming buyer at at
the limit price superior to the ask price, i.e., the lowest limit
price in the sell order book. As the result, the equilibrium of
\textit{ clearing price} process is always defined.

Since the matching mechanism does not add any information to the
economy, all information about asset prices is included in the order
flow. Whether public exchanges should or should not reveal the
real-time data contained in the order book is an important issue,
which continues to preoccupy the financial markets community
\cite{WisWang:02}. Our theoretical framework accommodates either
viewpoint, but our empirical application is better tailored to the
viewpoint that order books are public information. The current
blossoming of the trading activity \cite{Engle:00},
\cite{BollLitvTauch:2006}, \cite{Haut:08}, \cite{StoikAvell:08},
\cite{BanRussZhu:09} seems to confirm our viewpoint that traders are
(i) interested in understanding order book information, and (ii) trade
on that information.

We do not address in this paper the issue of differential information. The
market microstructure literature (such as \cite{Kyle:85}, and all following
models) considers various models of trading involving uninformed traders,
called noise traders, and one or several informed traders. One of the key
results of the Kyle model is that, given the information available to
the noise
traders, the resulting price process is a martingale in the appropriate
measure, whereas it may not be for the informed traders. As a consequence we
do not believe that abstracting issues of differential information is
limiting. The order books reflect all the public information. We will
show (under certain conditions) that the clearing price is a martingale in the
filtration corresponding to the public information, and not private
information.

There are roughly two different class of models in the liquidity
literature.  The first class of models (\cite{Jarr:92},
\cite{Jarr:94}, \cite{PlatSch:98}, \cite{PapanSirc:98},
\cite{Frey:98}, \cite{SchWilm:00}, \cite{BankBaum:04},
\cite{RogSin:10}) considers the action of a large trader who can
manipulate the prices in the market.  There are mainly two different
types of strategies a large trader can employ to that effect. The
first one is to corner the market, and then squeeze the shorts. The
second one is to "front-run one's own trades". While some exchanges
have rules to curtail the cornering of the market, the front-running
seems more difficult to ban from an exchange. It is known in
discrete-time trading, that, if there is no possibility of arbitrage
for small traders in periods where the large investor does not trade,
then there is no market manipulation strategy. In this paper, we
assume that Jarrow's conditions (\cite{Jarr:94}) for the absence of
the market manipulation strategy in discrete time hold. We exploit
then the theoretical results from \cite{BankBaum:04}, and
\cite{KallRhein:09} to prove the absence of arbitrage in our
continuous-time model, under certain conditions.

The second class of models (\cite{CetinJarrProt:04}, \cite{CetRog:07}, 
\cite{CetSonTou:10}, \cite{GokaySoner:11}
abstracts the issues of the market manipulation away, and considers all traders
as price-takers. In particular, \cite{CetinJarrProt:04} introduced an exogenous
residual supply curve against which an investor trades. The investor trades
market orders, and his/her order is matched instantaneously. As a
consequence of the instantaneity, it is plausible for \cite{CetinJarrProt:04} to assume the
"price effect of an order is limited to the very moment when the order is
placed in the market" (dixit \cite{BankBaum:04}), so that that the residual
supply curve at a future time is statistically independent from the order
just matched. For us however, since all information is contained in the
order flow, this assumption is not plausible, since it does not explain how
prices can incorporate information from the arrival of new orders.

The key to our approach is to distinguish between what we call
\textit{the cross
orders}, i.e., the orders submitted at a limit price such that they are likely to
be instantaneously matched, and what we call \textit{the uncross orders}, i.e.,
orders which will spend a positive amount of time in the order books before
being either matched or cancelled. We are not aware of workable assumptions
in the literature that warrant the absence of arbitrage strategy
involving the 
uncross orders. For cross orders however the task is simpler, and our
approach takes advantage of the results presented in \cite{BankBaum:04},
which assume that trading is immediate. Anecdotal evidence shows however
that market manipulation strategies tend in practice to occur over short
periods of time, casting doubts on the practicality of implementing market
manipulation strategies with uncross orders.

In our model all the information is contained in a Brownian sheet. The
Brownian sheet drives the dynamics of the demand curve. The advantage of
using a Brownian sheet is that we have the same cardinality of independent
sources of noise (namely, the cardinality of a real interval) as the
cardinality of the set of exogenous stochastic processes. Our methodology
generalizes the nonlinear partial differential equation approaches contained
in the literature cited above, in the same way that the Heath-Jarrow-Morton
methodology generalizes term structure models.

The main result in our article is to prove that if there exists a
unique equivalent martingale measure $\mathbb{Q}$ for the clearing
price, then (under some mild conditions) there is no arbitrage. We use
the Ito-Wentzell formula to obtain that result, and also to
characterize the dynamics of the demand curve and of the clearing
price in the measure $\mathbb{Q}$. We find that the volatility of the
clearing price is (up to a stochastic factor) inversely proportional to the sum of
the buy and sell order flow densities (evaluated at the clearing
price), which confirms the intuition that the volatility is inversely
proportional to the volume. We also demonstrate that our approach is
implementable. Although we do not prove a second fundamental theorem
of asset pricing, we naively use a special parameterization of our
model to price options. As in the early days of the
Heath-Jarrow-Morton methodology, we use only historical estimation (in
our case, of the order book) to fit our model and to solve for the
market price of risk. We expect that, should this paper meet with
interest around practitioners, market-implied implementations will see
the day.  Unsurprisingly, we obtain a smile curve for the implied
volatility. We note that this particular feature is not a very strong
sign of the adequacy of our approach to model asset prices, as most
models that came after \cite{BlckSch:73} result in a smile curve for
the implied volatility. Although limited, our results are however
encouraging. They show that a fairly demanding theoretical model can
be easily implemented.

The structure of the paper is as follows. In Section \ref{sec:model}, we introduce two
classes of models, one with atomistic traders only, and one with atomistic
traders and a large trader. In the first one, almost by definition, the
demand curve turns out to be continuous in time. In the second one,
continuity (under a set of assumptions) is proved. For both models, we
prove that there exist one or several martingale measures $\mathbb{Q}$ for
the clearing price and that under these measures there is no arbitrage. In
Section \ref{sec:char}, we characterize the price process under the risk-neutral measure,
giving the conditions under which it is unique. In Section \ref{sec:empirical}, we describe our
data set and the methodology we used to extract the relevant
information. We also present
our implementation, namely a simulation of option prices under the
risk-neutral measure.

\section{Model}\label{sec:model}
 
We assume a filtered probability space ($\Omega ,\mathcal{F},\{\mathcal{F}%
_{t}\},\mathbb{P}$). We will see later what processes generate the
filtration.

\bigskip

\subsection{The Market Mechanism}

\noindent A buy\ limit order specifies how many shares a trader wants to
buy, and at what maximum price he is willing to buy them. We call this price
the (buy) \textit{limit price}. A buy\ limit order specifies how many shares
a trader wants to buy, and at what maximum price he is willing to buy them.
We call this price the (sell)\ \textit{limit price}. Both buy and sell limit
prices are denoted by $p$, and should not be confused with the clearing
price, which is denoted by $\pi (t)$. The unmatched buy and sell orders are
kept in order books until they are either cancelled or matched with an
incoming order. An incoming order is matched with the order in the opposite
side of the market which has the best price. The clearing price of the
transaction is equal to the limit price of the order in the book, and not of
the incoming order. Partial execution is allowed, and ties are
resolved by the
time-priority. Below we present an example of the matching mechanism in
discrete time, that is at most one order arrives at time $t\in \{0,1,2,..\}.$

\begin{example}

Suppose that the clearing price at time $0$ is any price $\pi (0)\in \lbrack
100,120]$. After clearing, that is, when $0<t<1$ we suppose that the order
book contains the following orders

\begin{equation*}
\begin{tabular}{|l|l|}
\hline
\multicolumn{2}{|l|}{Buy Order Book} \\ 
Price & Quantity \\ \hline\hline
100 & 10 \\ \hline
\end{tabular}%
\text{ \ \ \ \ \ }%
\begin{tabular}{|l|l|}
\hline
\multicolumn{2}{|l|}{Sell Order Book} \\ 
Price & Quantity \\ \hline\hline
120 & 10 \\ \hline
130 & 10 \\ \hline
\end{tabular}%
\end{equation*}

At time $t=1$ a buy order arrives with a limit price of \$125, and a
quantity of 15. The exchange matches it with the best sell order, i.e., the
one with a sell limit price of \$120. However, the execution is only partial,
and the remainder of the buy order is placed in the order book at the limit
price of \$120, resulting in the following order book:

\begin{equation*}
\begin{tabular}{|l|l|}
\hline
\multicolumn{2}{|l|}{Buy Order Book} \\ 
Price & Quantity \\ \hline\hline
100 & 10 \\ \hline
125 & 5 \\ \hline
\end{tabular}%
\text{ \ \ \ \ \ }%
\begin{tabular}{|l|l|}
\hline
\multicolumn{2}{|l|}{Sell Order Book} \\ 
Price & Quantity \\ \hline\hline
130 & 10 \\ \hline
\end{tabular}%
\end{equation*}

The clearing price at time 1 is equal to the limit price of the sell order,
i.e.:

\begin{equation*}
\pi (1)=120.
\end{equation*}
\end{example}

This example illustrates several properties of the limit order markets. First,
the clearing price is always defined, and can assume any positive value
\footnote{%
We do not consider markets for swaps, where the prices can be negative.}.

Second, it is not inconceivable that an incoming order "crosses" the order
book, i.e., for the case of a buy order, that it best higher limit price
than the best sell order limit price, or best ask, since the buyer does not
lose a cent. Crossing the book is indeed advantageous for two reasons:
first, it allows for faster execution. In our example, had the buyer
submitted an order at price \$130 he would have bought the complete quantity
of shares (15)\ that he desired, rather than waiting an indeterminate amount
of time until enough sell orders arrive at his limit price. Second, suppose
that several buy orders are submitted at the same time. In case the demand
exceeds the supply at the best ask, the buy orders with the highest limit
price are executed first. Our own data analysis (see section \ref{sec:empirical}) shows
that few orders cross the NYSE ArcaBook \cite{arca:11}. This is consistent with the theory of
optimal order book placement suggested by Rosu \cite{Rosu:09}.
 
\begin{assumption}\label{as:1}
Buy and sell limit prices can assume any real value
between $0$ and $S$. They are usually denoted by $p$. Orders can be
submitted to the market at any time $t\in \mathbb{R}^{+}$.
\end{assumption}

\bigskip 

\subsection{The Brownian Sheet}

We now turn our attention to the continuous-time case. Note that we do
not prove convergence of a discrete-time model to our continuous-time
model. We take the latter as a given, plausible model of the
market. We start with a probability space $%
(\Omega ,\mathcal{F},\mathbb{P})$ satisfying the usual conditions. The
uncertainty is described by a one-dimensional Brownian sheet $W(t,s)$,
which is a continuous version of the construction (4.20) in
\cite{DaPratoZab:92} p. 100, i.e.:

\begin{equation*}
W(t,s)=\sum_{j=1}\beta _{j}(t)\int\limits_{0\leq \alpha \leq s}g_{j}(\alpha
)d\alpha, \text{ \ for\ }t\geq 0, 0\leq s\leq S
\end{equation*}

Here $\{\beta _{j}\}$ is a family of independent real-valued standard Wiener
processes, and $\{g_{j}\}$ is an orthonormal and complete basis for some
Hilbert space. The variable $t$ refers to time and the variable $s$
identifies the "factor" information necessary to model a large collection of
processes identified by the variable $p$, with $0\leq p\leq S$. While the
bound $S$ is usually taken to be equal to one in the literature, we assume
instead that $S$ is a large value for ease of modeling, as will become
clear later. The filtration $\{\mathcal{F}_{t}\}$ is generated by the
collection of Wiener processes $\{\beta _{j}\}$. A stochastic integral of
an $\mathcal{F}_{t}$-adapted integrand $\sigma (t,s)$ with respect to the
Brownian sheet is denoted by:

\begin{equation*}
I(T,S)=\int_{0}^{T}\int_{0}^{S}\sigma (t,s)W(dt,ds)
\end{equation*}

\subsubsection{A Market with Atomistic Traders}

\begin{definition}\label{def:1}
The net demand curve $Q$ is a function $[0,S]\times 
\mathbb{R}^{+}\times \Omega \mathbb{\rightarrow R}$, which value $%
Q(p,t,\omega )$ is equal to the difference between the quantity of shares%
\textbf{\ available} for purchase and the quantity of shares \textbf{%
available} for sale at price $p$ and at time $t$. For each $p$ the
stochastic process $Q(.,t,.)$ is $\mathcal{F}_{t}$-adapted. The random
variables $Q(p,t,.)$ are assumed to be uniformly bounded in $p$ and $t$ almost
surely.
\end{definition}

The net demand curve represents then the information available in the
(limit)\ order books of the exchange.

\begin{remark}\label{rem:1}
As we will see, the process $Q$ will be assumed to be continuous in time,
and we will prove that the clearing price is continuous in this model. Thus
we do not need to specify in this model whether $Q$ represents the net
demand just before clearing, or right after clearing, since these quantities
move continuously. The same remark will not apply a priori in our model with
a large trader.
\end{remark}

\begin{remark}\label{rem:2}
The net demand curve is decreasing in $p$.
\end{remark}

\begin{definition}\label{def:2}
  The clearing price $\pi (t)$ is an $\mathcal{F}_{t}$-adapted
  stochastic process which either satisfies
\begin{equation}
Q(\pi (t_{-}),t)=0,  \label{pi}
\end{equation}
when there is a solution to (\ref{pi}), or is otherwise defined by
continuation, i.e., $\pi (t)$ is equal to the value of $\pi $ at the latest
time $s<t$ for which there was a solution to $Q(\pi (s_{-}),s)=0$.
\end{definition}

When the net demand curve is continuous in time, we have of course,
for any $t$:

\begin{equation}
Q(\pi (t),t)=Q(\pi (t), t_{+}).
\end{equation}

\begin{definition}\label{def:3}
A limit order submitted at time $t$ \textit{crosses the
market} at time $t$ if:

\begin{itemize}
  \item either it is a buy order with limit price $p>\pi (t_{-})$
  \item or it is a sell order with limit price $p<\pi (t_{-})$
\end{itemize}

We call these orders\textit{\ cross orders}. All the other orders are called 
\textit{uncross orders}.
\end{definition}

In practice it is rare that limit orders cross the market, indeed there are
equilibrium models in which such orders are not rational \cite{Rosu:09}.

\bigskip

\begin{remark}\label{rem:3}
There is a more general way to model a market. We could
accumulate all of the buy order quantities with limit price higher than $p$
that enter the system (and withdraw the cancelled quantities when an order
is cancelled) into a cumulative demand curve $\mathcal{D}(p,t)$. Likewise we
could accumulate all of the sell order quantities with limit price less than $p$
that enter the system (and also withdraw the cancelled quantities) into a
cumulative supply curve $\mathcal{S}(p,t)$. We call these curves
"cumulative" because we include the order quantities that are matched in
them, and thus, if no order is cancelled, these curves increase in $t$. We
observe that:

\begin{equation*}
Q=\mathcal{D}-\mathcal{S}
\end{equation*}

Thus, modeling only the net demand curve is reductive. While there is
a (not very natural) way to extend the results of this paper to a
model including both the cumulated demand and supply curves, we
decided to present here only our simpler model for the following
reasons. As we argued in the previous section, orders rarely "cross",
there is a numerical instability when modeling both $\mathcal{D}$ and
$\mathcal{S}$.
\end{remark}

\begin{remark}\label{rem:4}
  Due to Assumption \ref{as:4}, the clearing price is limited to
  take values between $0$ and $S$. The frontier $\pi (t)=0$
  corresponds to a bankruptcy, and the frontier $\pi (t)=S$
  corresponds to a higher limit (say $%
  S=\$1M$)\ set by the exchange to prevent excessive speculation.
\end{remark}

\begin{remark}\label{rem:5}
Questions of uniqueness of the clearing price will be
addressed when we describe the stochastic differential equation that the
clearing price (should it exist) satisfies.
\end{remark}

The next assumption is standard. Note that by "transaction costs" we do not
mean the liquidity cost incurred, but an additional cost per transaction that
the exchange would charge the traders. For instance, a buyer would pay per
share a cost $\pi (t)+c$, with $c>0$.

\begin{assumption}\label{as:2}
The market is frictionless, i.e. $c=0$.
\end{assumption}

We now develop a model for the order quantity specified in any
order. Modeling a net demand curve which is twice differentiable with
respect to price is clearly easier than modeling a discrete curve, and
we assume it. It can occur when there is an uncountably infinite
amount of traders, and traders are atomistic. We need however to rule
out a degenerate case when a subset of non-zero measure of all the
traders agree on the limit price, thus generating a discontinuity in
the demand. This can be justified by assuming differential information
amongst traders about the real value of the stock, as in the market
microstructure literature. Note that the atomistic traders do not need
to be noise traders for the market to be consistent. Indeed, we can
(but do not need to)\ assume that all traders know all the information
contained in the order books at all times. If they do know this
information, as specified in the introduction, the clearing price does
not bring any extra information as the traders can compute the
clearing price by themselves at all times \footnote{%
  In order to account for the delay in information processing, it
  would be more plausible to assume that the information available to
  all traders at time $t$ corresponds to $\mathcal{F}_{t-}$ \ and not
  $\mathcal{F}_{t}$ but this would complicate the model significantly,
  while yielding not much extra conceptual value.}. Making the order book public
information can be an important assumption in some settings. Indeed,
our empirical studies (see Section \ref{sec:empirical}) show that the
order flow is tightly concentrated around the current clearing price. This
confirms the obvious intuition that traders are strategic and choose
their limit price based on their (slightly delayed)\ knowledge of the
clearing price. In this paper we will fit a parametric model to the
market which exploits this intuition. Our model is fairly
robust when the net demand is modelled as a function of the difference in
price $p-\pi (t)$, but not robust when it is modelled (as discussed above) as a
function of the limit price $p$ without a reference to the whole order
book information.

\begin{assumption}\label{as:3}
There is a continuum of atomistic buyers and sellers
who trade on the market. The resulting net demand curve $Q$ is twice
differentiable in price $p$ and continuous in $t$. We assume that
\begin{align*}
&\left.\frac{\partial Q}{\partial p}\right|_{p=0} = \left.\frac{\partial Q}{\partial p}
\right|_{p=S}=0 \\
&\left.\frac{\partial Q}{\partial p}\right|_{p} <0, \text{ \ \ \ for }
0<p<S
\end{align*}
\end{assumption}

For the moment we just state the stochastic differential equation that
$Q$ satisfies, assuming that the resulting net demand curve satisfies
Assumption \ref{as:3}. We defer the task of showing examples where
these conditions are satisfied to the next chapter, since in the next
model the exact same conditions will have to apply, and it is more
convenient for the reader to have a single place with the exact
specification of the model. Thus
\begin{eqnarray}
dQ(p,t) &=&\mu _{Q}(p,t)dt+\sigma
_{Q}(p,t)\int_{s=0}^{S}b_{Q}(p,s,t)W(ds,dt), %
\text{ \ for \ } 0\leq p\leq S  \label{sde} \\
Q(p,0) &=&Q_{0}(p), \text{ \ for \ }0\leq p\leq S  \notag
\end{eqnarray}%

The coefficients $\mu _{Q}$, $\sigma _{Q}$, and $b_{Q}$ are
$\mathcal{F}_{t}$-adapted. For the moment we just assume that they
are such that the solution to (\ref{sde}) exists, is unique, and is
uniformly bounded in $p$ and $t$ almost surely. Besides, for every
$p$, the process $Q(p,.)$ is a semimartingale. Also, we enforce:%
\begin{equation*}
\int_{s=0}^{S}b_{Q}^{2}(p,s,t)ds=1, \text{ \ for every \ } p
\text{ and } t.
\end{equation*}%

\begin{definition}\label{def:4}
  A (trading) strategy $\theta =(\theta (t))$ is a semimartingale that
  represents a number of shares held by the investor at each point in
  time. If the strategy is self-financing (see e.g. \cite{BankBaum:04}
  for a definition)\ the process $\beta ^{\theta }$ representing the
  value of the cash account is uniquely defined.
\end{definition}

We refer the reader to \cite{Jarr:94} for a definition of market
manipulation strategies in discrete time.

\begin{definition}\label{def:5}
For every real-valued $x$ an inverse process $P(x,t)$ satisfies
\begin{equation}
Q(P(x,t),t)=x.  \label{inv}
\end{equation}
The process $P(x, t)$ is undefined when (\ref{inv})\ does not admit a solution. 
\end{definition}

\begin{remark}
Since $Q$ is strictly monotonic in $p$, then whenever $Q$ exists, it is also
unique.
\end{remark}

\begin{definition}[(3.1) in \cite{BankBaum:04}]\label{def:6}
The asymptotic liquidation proceeds $L(\vartheta ,t)$
are defined as:

\begin{equation*}
  L(\vartheta ,t)=\int_{0}^{\vartheta }P(x,t)dx.
\end{equation*}%
\end{definition}

This definition is from \cite{BankBaum:04} for the proceeds of a fast
liquidation strategy of a large trader from $\vartheta $ to $0$. The
intuition behind this process will become more clear when we consider
a market with a large trader.

\begin{definition}\label{def:7}
The real wealth process achieved by a self-financing
trading strategy $\theta $ is given by
\begin{equation*}
  V^{\theta }(t)=\beta ^{\theta }(t)+L(\theta(t), t)
\end{equation*}
\end{definition}

\begin{lemma}[Lemma 3.2 in \cite{BankBaum:04}]
  For any self-financing semimartingale strategy $\theta$, the
  dynamics of the real wealth process $V^{\theta }$ are given by
\begin{align}
  &V^{\theta }(t)-V^{\theta }(0_{-}) \label{eq:bb04} \\
  &\quad = \int_{0}^{t}L(\theta (u_{-}),du)-\frac{1}{2
  }\int_{0}^{t}P^{\prime }(\theta (u_{-}),u)d[\theta ,\theta
  ]_{s}^{c}-\sum_{0\leq u\leq t}\int_{\theta (u_{-})}^{\theta
    (u)}\{P(\theta (u),u)-P(x,u)\}dx.\nonumber
\end{align}
\end{lemma}

\begin{definition}\label{def:8}
An arbitrage (strategy) is a self-financing trading
strategy $\theta $ such that $V^{\theta }(0_{-})=0$ and
\begin{eqnarray*}
\mathbb{P}(V^{\theta }(t) &>&0)>0, \\
\mathbb{P} (V^{\theta }(t) &>&0)\geq 0.
\end{eqnarray*}
\end{definition}

\begin{theorem}\label{thm:1}
Suppose in addition to our standing assumptions that

C1) for self-financing strategies involving only cross orders, Jarrow's
\cite{Jarr:94} discrete-time conditions for absence of market manipulation strategy
hold,

C2) no arbitrage strategy involves uncross orders,

C3) the volatility $\sigma _{Q}(p,t)$ is bounded away from zero, uniformly
in $p$,

C4) there is no path such that $Q(S,t)\geq 0$ or $Q(0,t)\leq 0$.

\noindent Then

F1) there exists at least one martingale measure $\mathbb{Q}$ for $\int
L(\vartheta ,dt)$,

F2) there is no arbitrage strategy,

F3) the clearing price $\pi (t)$ is continuous,

F4) any such measure $\mathbb{Q}$ is also a martingale measure for $\pi (t)$.
\end{theorem}

\begin{proof}
All arguments of the proof of Theorem \ref{thm:2} can be used,
replacing the (non-infinitesimal)\ strategy $\theta $ of the large trader by
an infinitesimal strategy $\theta dp$ of an atomistic trader. Since Theorem
\ref{thm:2} is more general, we will only show the proof of Theorem \ref{thm:2}.
\end{proof} 

A market composed only of atomistic traders is not very realistic. In the
next section, we will show that under certain conditions a large trader can be
added to our market, and the no-arbitrage conditions can still be obtained.

\subsection{A Market with Atomistic Traders and a Large Trader}

A continuous (in time) net supply curve can also arise when large traders
decide to split their large orders into atomistic orders. \cite{CetinJarrProt:04} show that this is indeed an optimal
strategy in their model. The same decision turns out to be optimal in our
model, under certain conditions which are elaborated in the next
assumptions. We will prove this fact in Theorem \ref{thm:2}, thus making our model is
self-consistent.

Compared to \cite{CetinJarrProt:04}, there are two main conceptual difficulties when adding a
large trader into the model. First, the trader may submit orders which are
not instantaneously matched, unlike the market orders in the \cite{CetinJarrProt:04} model.
This is why we consider separately in this section the cross orders and the
uncross orders. The cross orders will be matched instantaneously, and thus
they behave like market orders in the \cite{CetinJarrProt:04} model, and continuity (in time)
can be proved. For uncross orders, one would need a more detailed economic
model to specify in which cases it is advantageous for traders to submit a
continuous (in time) order flow. Rather than going into these details, we
assume that this holds true. In any event, market manipulation strategies
(such as market cornering) are probably more likely to be implemented with
cross orders than with uncross orders.

The second difficulty is that we have to prove that the large trader
cannot manipulate the market. The conditions in \cite{Jarr:94} apply when
traders can submit orders at discrete time intervals. \cite{BankBaum:04}
and \cite{KallRhein:09} show the conditions under which the
discrete-time conditions extend to continuous time, and our proof of
consistency of the market with a large trader will consist in part in
checking that \cite{KallRhein:09} conditions apply.

Since the continuity in time of the net demand is not assumed any more
(see Remark \ref{rem:1}), we need now to distinguish between the
incoming orders (or the order flow) of the large trader, and the order
book position.

\begin{definition}\label{def:9}
The net demand curves of a large (atomistic)
trader $Q_{L}$ ($Q_{A}$) is a function $[0,P]\times \mathbb{R}^{+}\times
\Omega \mathbb{\rightarrow R}$ whose value $Q_{L}(p,t,\omega )$ \ ($%
Q_{A}(p,t,\omega )$)\ is equal to the difference between the quantity of
shares \textbf{submitted} for purchase and the quantity of shares \textbf{%
submitted} for sale at price $p$ at time $t$. For each $p$ the stochastic
processes $Q_{L}(.,t,,)$ and $Q_{A}(.,t,.)$ are $\mathcal{F}_{t}$-adapted
semimartingales. As before the net demand of the atomistic traders satisfies
\begin{eqnarray}
  dQ_{A}(p,t) &=&\mu _{Q_{A}}(p,t)dt+\sigma
  _{Q_{A}}(p,t)\int_{s=0}^{S}b_{Q_{A}}(p,s,t)W(ds,dt)\text{ \ for \ } 0\leq p\leq
  S, \\
  Q_{A}(p,0) &=&Q_{A,0}(p) \text{ \ for \ } 0\leq p\leq S.  \notag
\end{eqnarray}%
\end{definition}

\begin{definition}\label{def:10}
For every real-valued $x$ the process $P_A(x,t)$ satisfies
\begin{equation}
  Q_{A}(P_A(x,t),t)=x.
\end{equation}%
\end{definition}

\begin{definition}\label{def:11}
The asymptotic liquidation proceeds of the large trader $%
L_{L}(\vartheta ,t)$ are defined by
\begin{equation*}
  L_{L}(\vartheta ,t)=\int_{0}^{\vartheta }P_A(x,t)dx.
\end{equation*}%
\end{definition}

\begin{assumption}\label{as:4}
Both $Q_{L}$ and $Q_{A}$ are twice differentiable in $p$. Only 
$Q_{A}$ is assumed to be continuous in $t$.
\end{assumption}

\begin{remark}\label{rem:6}
The (total) net demand curve satisfies
\begin{equation*}
  Q=Q_{L}+Q_{A}.
\end{equation*}
\end{remark}

\begin{assumption}\label{as:5}
For simplicity we assume%
\begin{equation*}
  Q(0,t)>0.
\end{equation*}
\end{assumption}

\begin{assumption}\label{as:6}
For each $p\geq \pi (t)$, the function $Q_{L}(p,t)$ is
continuous in time.
\end{assumption}

\begin{theorem}\label{thm:2}
Suppose in addition to the standing assumptions that

C1) for self-financing strategies involving only cross orders,
Jarrow's \cite{Jarr:94} discrete time conditions for absence of market
manipulation strategy hold,

C2) no arbitrage strategy involves uncross orders,

C3) the volatility $\sigma _{Q_{A}}(p,t)$ is bounded away from zero,
uniformly in $p$,

C4) there is no path such that $Q(S,t)\geq 0$ or $Q(0,t)\leq 0$.

\noindent Then

F1) there exists at least one martingale measure $\mathbb{Q}$ for $\int
L_L(\vartheta ,dt)$,

F2) there is no arbitrage strategy,

F3) the net demand curve $Q$ is continuous in $t$,

F4) the clearing price $\pi (t)$ is continuous,

F5) any such measure $\mathbb{Q}$ is also a martingale measure for $\pi (t)$.
\end{theorem}

\noindent The full proof is presented in the Appendix. Here we show
only a summary of the proof.

\vspace{5mm} {\it Summary of the Proof: } We verify that the
\cite{KallRhein:09} conditions hold. Thus $\int L_L(\vartheta ,dt)$ is a
$\mathbb{Q}$-martingale, and no market manipulation exists that
involves cross orders.  We prove that tame strategies ($\theta $
continuous in $t$) are optimal for the large trader so that, first $Q$
is continuous in $t$., and, second, cross orders are traded only at
the clearing price $\pi (t)$.
\begin{flushright}
  $\qedsymbol$
\end{flushright}

\begin{remark}\label{rem:7}
Under certain conditions it is possible to extend our
model to more than one large trader on the market. Following the assumptions
of Theorem \ref{thm:2}, the net supply curve $Q$ in a market with one large trader is
indistinguishable from the net supply curve in a market with only atomistic
traders. It is thus plausible that, should a second large trader arrive in
the market, he would behave like the first large trader. However, there are
many different ways to justify this result, and we believe that this
discussion would be more appropriate in an economics journal than here.
\end{remark}

\section{Characterization of the Price Process in the Risk-Neutral
  Measure}\label{sec:char}

In this section we characterize the price process in either one of the
market models we specified earlier, since they turned out to be equivalent
for our purposes.

\begin{assumption}\label{as:7}
  To avoid repetition we assume in this section that there is no path
  such that $Q(S,t)\geq 0$ or $Q(0,t)\leq 0$.
\end{assumption}

\begin{definition}\label{def:12}
The market price of risk $\lambda $ is a function $%
[0,S]\times \mathbb{R}^{+}\times \Omega \mathbb{\rightarrow R}$. The market
price of risk process $\lambda (s,.,.)$ is an $\mathcal{F}_{t}$-adapted
semimartingale for every $s\in \lbrack 0,S].$
We define the $\mathbb{Q}$-measure as a measure such that the process $W^{%
\mathbb{Q}}$ is a Brownian sheet, where:%
\begin{equation}
  W^{\mathbb{Q}}(ds,dt)=W(ds,dt)+\lambda (s,t)dt.  \label{WQ}
\end{equation}
We can then define the clearing price process as
\begin{equation*}
  d\pi (t)=\sigma _{\pi }(t)\int_{s}b_{\pi }(s,t)W^{\mathbb{Q}}(ds,dt),
\end{equation*}
with
\begin{equation}
  \int_{s}b_{\pi }(s,t)^{2}ds=1.  \label{factob}
\end{equation}
\end{definition}

Since $Q(p,t)$ must be strictly decreasing in $p$, we find it
convenient to slightly modify the definition into:

\begin{equation*}
  Q(p,t)=Q(p,0)-\int_{0}^{p}q(y,t)dy,
\end{equation*}

where we define $Q(0,t)$ and $q(p,t)$ (for $0<p\leq S$) to be strictly
positive processes with
\begin{align}
  dQ(0,t) &=\mu _{Q}(0,t)dt-\sigma
  _{Q}(0,t)\int_{s}b_{q}(0,s,t)W(ds,dt)
  \quad Q(0,0)=Q_{0}(0),  \label{SDE1} \\
  dq(p,t) &=\mu _{q}(p,t)dt+\sigma
  _{q}(p,t)\int_{s}b_{q}(p,s,t)W(ds,dt) \quad
  q(p,0)=Q_{0}(p),  \label{SDE2} \\
  &\text{ \ for }0<p\leq S \nonumber\\
  q(0,t) &=0.  \nonumber
\end{align}%

Like before, the coefficients of the equations (\ref{SDE1}) and (\ref{SDE2})
are $\mathcal{F}_{t}$-adapted. We assume that the solution to (\ref{SDE1})
and (\ref{SDE2}) exists, is unique, and is uniformly bounded in $p$ and $t$
almost surely. Also, we require that
\begin{equation*}
  \int_{s=0}^{S}b_{q}^{2}(p,s,t)ds=1 \text{ \ for every \ } p \text{\
    and \ } t.
\end{equation*}

\begin{remark}\label{rem:8}
The process $q$ is a density of orders. By definition
\begin{eqnarray*}
q(p)dp &=&\text{quantity of shares available for purchase \ with limit price
in }[p,p+dp]\text{ } \\
&&\text{ + quantity of shares available for sale with limit price in }[p,p+dp].
\end{eqnarray*}%
\end{remark}

\begin{remark}\label{rem:9}
  Assuming that $Q(0,t)$ is twice-differentiable in $p$, we must make
  sure that $q(p,t)$ is differentiable in $p$, for the process
  $Q(p,t)$ to be twice-differentiable in $p$. This occurs if, for
  instance,
\begin{equation*}
dq(p,t)=\int_{s=0}^{p}(p-s)W(ds,dt).
\end{equation*}
\end{remark}

We now define the following processes:
\begin{eqnarray*}
C(\pi ,t) &=&-\sigma _{\pi }(t)\left( \frac{\partial }{\partial p}\left(
\sigma _{Q}(0,t)\int_{s}b_{q}(0,s,t)b_{\pi }(s,t)ds\right) +\sigma _{q}(\pi
(t),t)\int_{s}b_{q}(\pi ,s,t)b_{\pi }(s,t)ds\right),  \\
b(\pi ,t) &=&-\mu _{Q}(0,t)+\int_{0}^{\pi }\mu _{q}(p,t)dpdt+\frac{1}{2}%
\frac{\partial q}{\partial p}(\pi ,t)(\sigma _{\pi }(t))^{2}-C(\pi ,t), \\
\Sigma (\pi ,s,t) &=&\int_{0}^{\pi }\sigma _{q}(p,t)b_{q}(p,s,t)ds
\end{eqnarray*}

\begin{remark}\label{rem:10}
  If one is not interested in modeling the correlation between orders
  at a different limit price, one may think that a Brownian sheet is
  not necessary, in the same way that the Heath-Jarrow-Morton model is
  often implemented with only 2 factors\footnote{ Note however that
    Carmona and Tehranchi \cite{CarTer:03} show that the
    Heath-Jarrow-Morton model, if there are less factors than forward rates,
    may result in a special type of arbitrage.}. As we shall see
  however, the full complexity of a Brownian sheet is necessary for
  the market price of risk equations to have a solution.
\end{remark}

\begin{definition}\label{def:13}
The market price of risk equations are:
\begin{equation*}
  \int_{s=0}^{P}\Sigma (\pi ,s,t)\lambda (s,t)ds=b(\pi ,t),\text{ \ for \ }
  0\leq \pi \leq P.
\end{equation*}
\end{definition}

\begin{theorem}\label{thm:3}
Suppose that the previous assumptions hold true. In addition,
suppose that the market price of risk equations have a unique solution. Then
there is no arbitrage.
\end{theorem}
 
\begin{proof}
  A market clears if $Q(p(t),t)=0$, or, equivalently, if
  $dQ(p(t),t)=0$. We use the Ito-Wentzell formula to compute
  $dQ(p(t),t)$ and set $dQ(p(t),t)=0$.

\begin{align}
  \mu _{Q}(0,t)dt&-\int_{0+}^{\pi (t)}\mu _{q}(p,t)dpdt-\int_{0}^{\pi
    (t)}\sigma _{q}(p,t)\int_{s}b_{q}(p,s,t)W(ds,dt)dp  \label{ITOWE} \\
  &-q(\pi (t),t)\sigma _{\pi }(t)\int_{s}b_{\pi
  }(s,t)W^{^{\mathbb{Q}}}(ds,dt)-%
  \frac{1}{2}\frac{\partial q}{\partial p}(\pi (t),t)(\sigma _{\pi
  }(t))^{2}dt+C(\pi (t),t)dt=0. \notag
\end{align}

\noindent We equate the volatility terms to zero above and find
\begin{equation}
  \sigma _{\pi }(t)b_{\pi }(s,t)=-\frac{\int_{0}^{\pi (t)}\sigma
    _{q}(p,t)b_{q}(p,s,t)dp}{q(\pi (t),t)}.  \label{vovol}
\end{equation}

\noindent We substitute (\ref{WQ})\ in (\ref{ITOWE})\ and equate the drift terms to
zero. This results in
\begin{align*}
  \int_{s}\int_{0}^{\pi (t)}&\sigma _{q}(p,t)b_{q}(p,s,t)dp\lambda (s,t)ds= \\
  &\quad-\mu _{Q}(0,t)+\int_{0+}^{\pi (t)}\mu
  _{q}(p,t)dpdt+\frac{1}{2}\frac{
\partial q}{\partial p}(\pi (t),t)(\sigma _{\pi }(t))^{2}-C(\pi (t),t).
\end{align*}

\noindent Since the above must hold for any value of $\pi (t)$, then
the market price equations must be satisfied. Thus there exists a
unique measure $\mathbb{Q}$ such that $\pi $ is a martingale. By
theorems \ref{thm:1} and \ref{thm:2}, there exists a non-empty set
$\mathcal{Q}$ of martingale measures for $\int L(\vartheta
,dt)$. Besides, any measure $\mathbb{\tilde{Q}\in }$ $\mathcal{Q}$
must be a martingale measure for $\pi $. Uniqueness of $\mathbb{Q}$
ensures that $\mathcal{Q=\{}\mathbb{Q\}}$, thus $\mathbb{Q}$ is a
martingale measure for $ \int L(\vartheta ,dt)$. Therefore Theorems
\ref{thm:1} and \ref{thm:2} imply no arbitrage.
\end{proof}

\begin{remark}\label{rem:11}
In a numerical implementation, the market price of risk
equations will be a set of $S$ linear equations with $S$ linear unknowns.
Generically, like other market price of risk equations in finance, these
will admit a unique solution. This was the case in all our simulations.
\end{remark}

\begin{remark}\label{rem:12}
Integrating (\ref{vovol}) and using (\ref{factob}), we see
that
\begin{equation}
  \sigma _{\pi }(t)=\frac{(\int_{0}^{S}(\int_{0}^{\pi (t)}\sigma
    _{q}(p,t)b_{q}(p,s,t))^{2}dpds)^{1/2}}{q(\pi (t),t)}.  \label{tagad}
\end{equation}

\noindent The denominator of (\ref{tagad})\ shows that more orders at the clearing
price decrease volatility. This is to be expected. The effect of the
numerator is harder to analyze, and shows that the volatility of the whole
net demand curve affects the volatility of the clearing price.
\end{remark}

Before we move to the implementation of our model, we perform an empirical
analysis of the market, which will guide us in specifying a parametric model.

\section{Empirical Analysis}\label{sec:empirical}

\subsection{Implementation}

To the best of our knowledge, our methodology to model prices is quite
original. In the same way that the HJM\ methodology opened the way to
the development of several parametric models, we hope that this paper
will result in an effort to develop models for the demand curve that
are appropriate for risk management. This section only scratches the
surface of that effort. Alternatively, this is what practitioners would
call a "proof of concept", showing that the model can be implemented
without too much effort. We do not make any claims about the power of
the model we present hereafter, which is perhaps the simplest
no-arbitrage model imaginable in our framework.

Compared to the Heath-Jarrow-Morton model for the forward rates, it is
necessary in our model to know the drift of $Q(p,t)$ in the physical
measure in order to simulate the price in the risk-neutral measure. In
practice, there are two different methods to determine this drift:\
historical estimation or market implied. In the historical estimation
method, one calibrates the physical drift $\mu _{Q}(0,t)$ in the model
to market observables, and then solves for the market price of
risk. In the market implied method, one calibrates the model to, say,
a smile curve of option prices, like for other stochastic volatility
models. This paper fell short of proving a second fundamental theorem
of asset pricing, which would justify the implied market method. This
will be a goal a subsequent paper.

In our implementation, we first use the historical estimation method to
estimate the drift of $Q$ in the physical measure, and then
calculate options prices by simulation, in essence assuming that the second
fundamental theorem works.

\subsection{A\ Model with Relative Prices}

As explained in the introduction, and shown later, a model will be
more robust if it assumes a direct dependence between the net demand
curve $Q(p,t)$ and the relative prices $p-\pi (t)$. This takes into
account the fact that the investors observe the market in real-time.
We are then obligated to introduce $Q(p,\pi ,t)$ as the net demand
curve at price $p$ when the clearing price is $\pi $. We have then as
before
\begin{eqnarray}
  Q(p,\pi (t),t) &=&Q(0,\pi (t),t)-\int_{0}^{\pi (t)}q(p,\pi (t),t)dp,
  \label{BIGREL} \\
  dq(p,\pi (t),t) &=&\mu _{q}(p,\pi (t),t)dt+\sigma _{q}(p,\pi
  (t),t)\int_{s}b_{q}(p,\pi (t),s,t)W(ds,dt),  \notag \\
  dQ(0,\pi (t),t) &=&\mu _{Q}(0,\pi (t),t)dt-\sigma _{q}(0,\pi
  (t),t)\int_{s}b_{q}(0,\pi (t),s,t)W(ds,dt).  \notag
\end{eqnarray}

\noindent Let
\begin{equation*}
  F(\pi ,t)=Q(0,\pi ,t)-\int\limits_{0}^{\pi }q(p,\pi ,t)dp.
\end{equation*}

\noindent Then
\begin{equation*}
  F_{\pi }(\pi ,t)=\frac{\partial Q(p,\pi ,t)}{\partial \pi }-q(\pi ,\pi
  ,t)-\int\limits_{0}^{\pi }\frac{\partial }{\partial \pi }q(p,\pi ,t)dp
\end{equation*}

\noindent and
\begin{equation*}
  F_{\pi \pi }(\pi ,t)=\frac{\partial ^{2}Q(0,\pi ,t)}{\partial \pi ^{2}}%
  -q_{p}(\pi ,\pi ,t)-q_{\pi }(\pi ,\pi ,t)-\frac{\partial }{\partial \pi }%
  q(\pi ,\pi ,t)-\int\limits_{0}^{\pi }\frac{\partial ^{2}}{\partial \pi ^{2}}%
  q(p,\pi ,t)dp.
\end{equation*}

\noindent We define
\begin{equation*}
H(s,\pi ,t)=-\int\limits_{p=0}^{\pi }\sigma _{q}(p,\pi ,t)b_{q}(p,\pi ,s,t)dp.
\end{equation*}

\noindent Hence
\begin{equation*}
\frac{\partial }{\partial \pi }H(s,\pi ,t)=-\int\limits_{p=0}^{\pi }\frac{%
\partial \sigma _{q}}{\partial \pi }[\sigma _{q}(p,\pi ,t)b_{q}(p,\pi
(t),s,t)]dp-\sigma _{q}(\pi ,\pi ,t)b_{q}(\pi ,\pi ,s,t).
\end{equation*}

\noindent Therefore the Ito-Wentzell's formula takes the form
\begin{eqnarray*}
dQ(\pi (t),t) &=&\mu _{Q}(0,\pi (t),t)-\int\limits_{p=0}^{\pi (t)}\mu
_{q}(p,\pi (t),t)dt-\sigma _{q}(p,\pi (t),t)\left( \int_{s}b_{q}(p,\pi
(t),s,t)W(ds,dt)\right) dp \\
&&+F_{\pi }(\pi (t),t)d\pi (t)+\frac{1}{2}F_{\pi \pi }(\pi (t),t)(d\pi
)^{2}+C(\pi (t),t)dt,
\end{eqnarray*}%

\noindent where
\begin{equation*}
C(\pi (t),t)=\sigma _{\pi }(t)\int \frac{\partial }{\partial \pi }H(s,\pi
(t),t)ds.
\end{equation*}

\noindent Market clears if $dQ(\pi (t),t)=0$. Equating the
volatilities results in
\begin{equation}
\sigma _{\pi }(t)b_{\pi }(s,t)=\frac{\int_{0}^{\pi (t)}\sigma _{q}(p,\pi
(t),t)b_{q}(p,\pi (t),s,t)dp}{F_{\pi }(\pi (t),t)}.  \notag
\end{equation}

\noindent Apart from the denominator, this is the same equation as (\ref{vovol}). We now
define
\begin{equation}
  b(\pi ,t)=\int_{0}^{\pi }\mu _{q}(p,\pi ,t)dp-\mu _{Q}(0,\pi ,t)-\frac{1}{2}%
  F_{\pi \pi }(\pi (t),t)\sigma _{\pi }^{2}(t)-C(\pi (t),t).
\label{G_withprice}
\end{equation}%

\noindent Finally, the market price of risk equations can be expressed as
\begin{equation}
  \int_{s}\left( \int_{0}^{\pi }\sigma _{q}(p,\pi ,t)b_{q}(p,\pi
    ,s,t)dp\right) \lambda (s,t)ds=b(\pi ,t).  \label{driftwithprice}
\end{equation}

\subsection{A Parametric Implementation with Relative Prices}

We discretize our model in 3 dimensions: time $t$, limit price $p$ and
factor $s$. For the market price of risk equations to have a unique
solution the number of price buckets where orders are assigned should
be equal to $S$, namely the number of factors. We let $\Delta p$ be
the size of a price bucket.

We define the relative price $k$ as
\begin{equation*}
  k\equiv p-\pi,
\end{equation*}

\noindent and the relative net demand $\tilde{Q}(k,\pi ,t)$ as
\begin{equation*}
  \tilde{Q}(k,\pi ,t)=Q(p,\pi ,t).
\end{equation*}

We simplify the model above (as expressed in \eqref{BIGREL}), and
assume that the order flow depends only on the relative price $k$,
i.e., $\tilde{Q}(k,\pi ,t)=$ $\tilde{Q}(k,t)$. Now assume without loss
of generality that $S$ is even, and define $K=S/2$. We model the
relative net demand curve as
\begin{equation*}
  \tilde{Q}(k,t)=\tilde{Q}(-K,t)+\sum_{l=-K+1}^{k}\tilde{q}(l,t).
\end{equation*}

In our model, the logarithm of the order flow quantities follow
Ornstein-Uhlenbeck processes. This ensures the stationarity as well as the
positivity of the order flow quantities. In other terms
\begin{align}
  d\log \tilde{Q}(0,t) &=-a_{Q}(0)(\log \tilde{Q}(0,t)-\log \hat{Q}%
  (0))dt+\sigma _{Q}^{rel}(0)\sum_{j=-K+1}^{K}b_{q}(k,j)\sqrt{\Delta
    p} dW_{j}(t)  \label{logbigQ} \\
  d\log \tilde{q}(k,t) &=-a_{q}(k)(\log \tilde{q}(k,t)-\log \hat{q}%
  (k))dt+\sigma _{q}^{rel}(k)\sum_{j=-K+1}^{K}b_{q}(k,j)\sqrt{\Delta
    p}
  dW_{j}(t)\text{, with }\label{logq}\\
  &k=-K+1..K \nonumber
\end{align}

It is thus necessary to estimate the parameters $a_q(k)\geq 0$,
$a_Q(0)\geq 0$, $\sigma _{q}^{rel}(k)$ and $b_{q}(k,j)$, as well as
the initial values. We notice that our model is not twice
differentiable in $p$, however, for our parameter values, it does not
behave significantly differently from a smoothed version of that
model.

\subsubsection{Results}

\paragraph{Data:}
We collected high frequency data from NYSE ArcaBook \cite{arca:11} for General
Electric Company (GE) on April 1$^{st}$, 2011. GE limit order is
characterized by six intrinsic quantities: (1) message type:
e.g. ``A": add new order,``M": modify order, ``D": delete order; (2)
trading type: e.g. ``B": buy limit order or ``S": sell limit order;
(3) time: recorded when a new event occurs, e.g. add new order, modify
order or delete order; (4) ID: a unique identifier of each limit
order; (5) price in dollars; and (6) size in
number of shares.

We selected data during the trading time from 9:30 AM to 4:00 PM
EST. Before doing any analysis, we removed obvious errors in the data,
e.g. abnormal prices of USD 122 or USD 0.01. We managed to keep track
of approximately 90\% of the original limit orders.
For the cancellation of limit orders, we assumed that if the amount of
time was less than or equal 2 minutes when the ``delete'' message
occurred after the ``modify''
message of the same order, this order was cancelled.\\

Let $p_{\text{min}}$ be the minimum price of USD 20.00 and
$p_{\text{max}}$ be the maximum price of USD 20.62 during the trading
day. We partitioned the price space into 5 cents bins
\begin{eqnarray*}
	p_{\text{min}} = 20.00 &=& p_{-K} < p_{-K+1} < \dots < p_{K-1}
        < p_K = p_{\text{max}} = 20.62\\
	p_{K+1} - p_K &=& \Delta p ~(5 ~\text{cent}).
\end{eqnarray*}

Likewise, we partitioned the time into:
\begin{equation*}
	t_{i+1} - t_i = \Delta t ~(1 ~\text{minute})
\end{equation*}

\paragraph{Clearing Prices $\pi(t)$:}
By dividing the time period from 9:30 AM to 4:00 PM EST into one-minute time intervals, we obtained 390 one-minute time intervals (6.5 trading hours * 60 minutes).

\begin{table}[htdp]
\begin{center}
\caption{Statistical Summary for $\pi(t)$}
\begin{tabular}{llcll}
\hline 
\textbf{Summary} & \textbf{Output} & &\textbf{Summary} & \textbf{Output}\\
\hline \hline
nobs & 390 & & SE Mean & 0.004174\\
\hline
NAs & 0 & & LCL Mean & 20.3294\\
\hline
Minimum & 20.05 & &UCL Mean & 20.3458\\
\hline
Maximum & 20.61 & &Variance & 0.006795\\
\hline
1. Quartile & 20.29 & &Stdev & 0.082435\\
\hline
3. Quartile & 20.39 & &Skewness & -0.289841\\
\hline
Mean & 20.338 & &Kurtosis & 0.277282\\
\hline
Median & 20.340 & &Sum & 7931.68\\
\hline  \hline
\end{tabular}
\end{center} 
\end{table}%

\subparagraph{Jarque-Bera Test:} We implemented the Jarque-Bera test
to check the null hypothesis H$_0$ that the clearing prices $\pi(t)$
are normally distributed (while H$_1$ is the hypothesis that $\pi(t)$
are not normally distributed). The Jarque-Bera test used both skewness
and kurtosis simultaneously to check the normality of the selected
data. The Jarque-Bera test statistic was calculated as
\begin{equation*}
	\text{JB} = (S^*)^2 + (K^*)^2,
\end{equation*}
where
\begin{equation*}
	S^{*} = \sqrt{\frac{T}{6}} \ \widehat{S} \left( \pi(t) \right)
        \sim N(0,1) \qquad \text{ and } \qquad
	K^* = \sqrt{\frac{T}{24}} \left( \widehat{K} \left( \pi(t) \right) - 3 \right) \sim N(0,1).
\end{equation*}

The Jarque-Bera test yielded the p-value of 3.161\%, which is below
the default significance level of 5\%. Thus, we rejected the null
hypothesis $H_0$ that the clearing prices $\pi(t)$ are normally
distributed. However, with the p-value of 3.161\%, the distribution of
$\pi(t)$ is rather close to a normal distribution, which is also shown
in Figure 1. 

\begin{figure}[htbp]
  \begin{center}
  \includegraphics[width=0.45\textwidth]{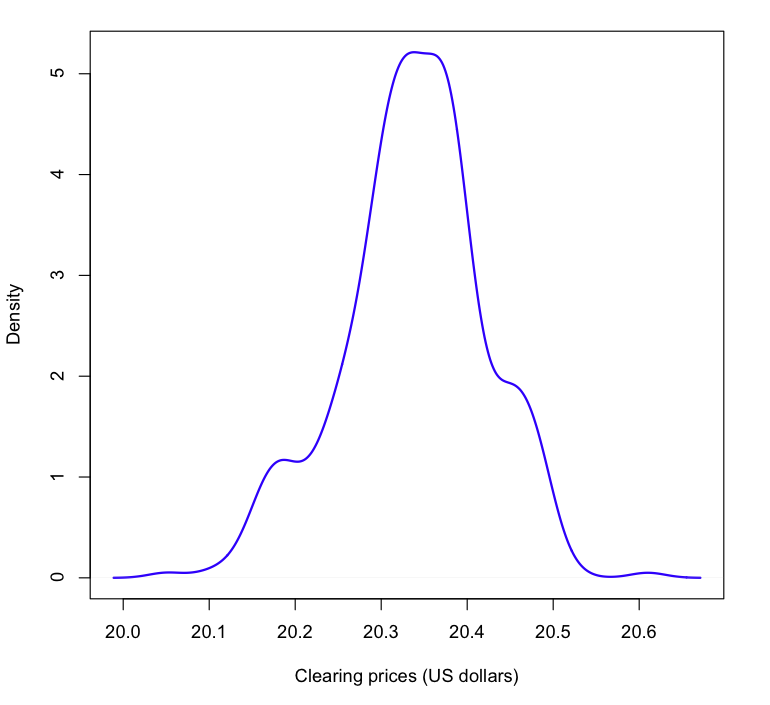}
  \caption{The density of clearing prices $\pi(t)$}
  \label{densityQQ_adjPi}
 \end{center}
\end{figure}

\subparagraph{Evolution of $\pi(t)$:}
From the result of linear regression, we obtained the constant value of the drift of clearing prices $\pi(t)$. Thererfore, we have
\begin{eqnarray*}
	d \pi(t) &=& cdt + \sigma_{\pi}(t) \int_s b_{\pi}(s,t) W(ds,dt),\\
	&& \qquad \qquad \qquad \qquad \text{where: } c = 0.000374764.
\end{eqnarray*}

\begin{figure}[!h]
\begin{center}
	\includegraphics[width=0.4\textwidth]{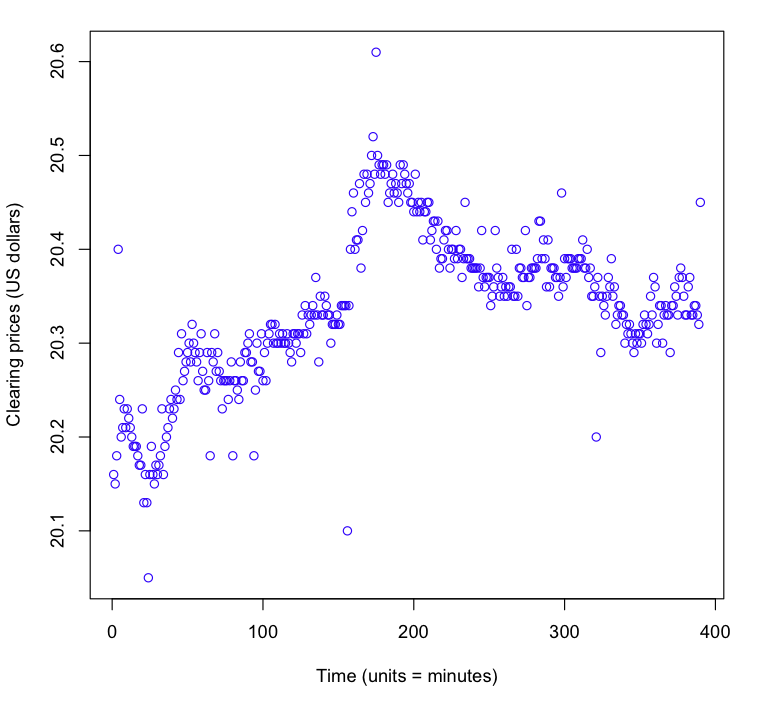}
	\includegraphics[width=0.4\textwidth]{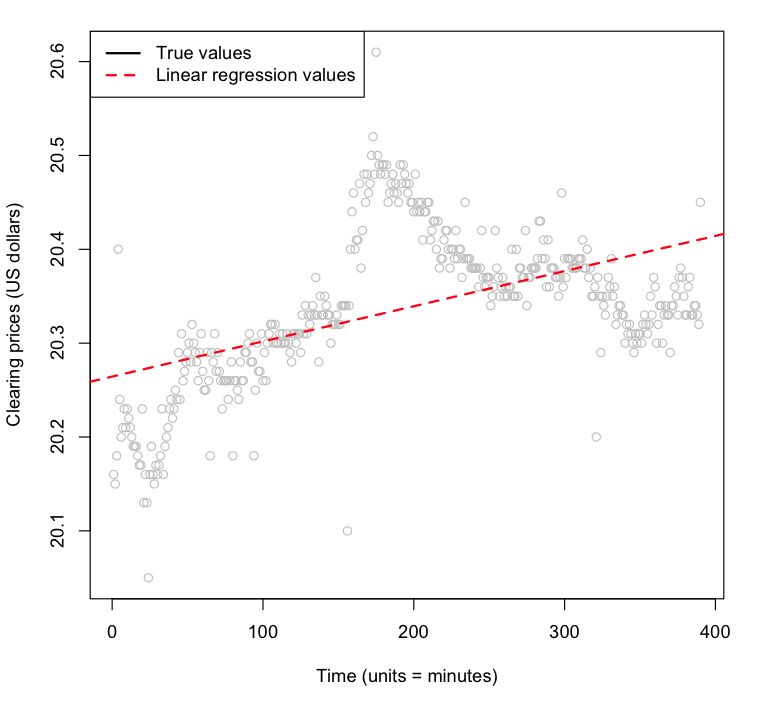}
\caption{GE clearing prices $\pi(t)$ on April 1$^{st}$, 2011}
\label{reg_true_plot}
\end{center}
\end{figure}

\paragraph{Buy Limit Orders:} Let us recall the definition of $\tilde{Q}$
\begin{eqnarray*}
	\log \tilde{Q}(-K, i \Delta t) &=& \log \left( \sum \text{number of buy orders} * \text{buy order quantity arriving before} \ i \Delta t \right),\\
	&& \qquad \qquad \text{where } \Delta t = 1\! \text{ minute}.
\end{eqnarray*}

We tested the data of $\log \tilde{Q}(i \Delta t)$ for one-minute time
interval with the autocorrelation and partial autocorrelation
functions.  The ACF and PACF figures showed (see Figure 3) that  the data of $\log \tilde{Q}(i \Delta t)$ for one-minute time interval were fitted in the autoregressive AR(1) model.
\begin{figure}[htbp]
\begin{center}
  \includegraphics[width=0.45\textwidth]{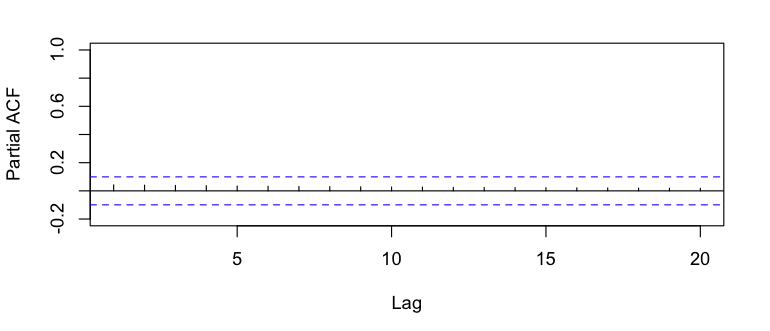}
  \includegraphics[width=0.45\textwidth]{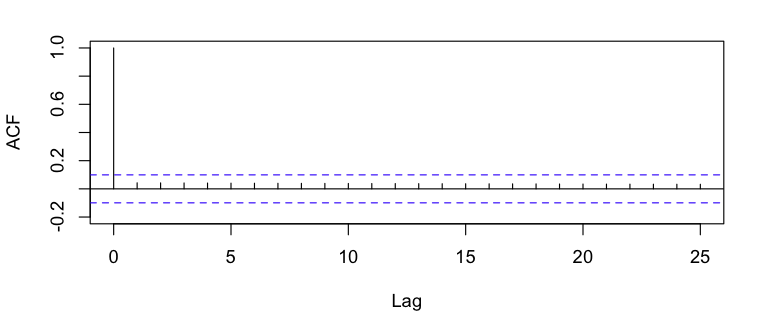}
\label{acf_logQbar_1min}
\caption{ACF and PACF of $\log \tilde{Q}(i \Delta t)$ for one-minute time interval.}
\end{center}
\end{figure}

\paragraph{Results:}
We conducted the estimation of our model (\ref{logbigQ}), (\ref{logq})\
over the data set described above. The opening clearing price on that day was
\begin{equation*}
	\pi (0)=20.16.
\end{equation*}%
Using $K=7$ and $\Delta p=0.05$, the speed of mean-reversion was obtained from the autoregressive AR(1) model. For the relative net demand curve, we observed the
following initial value and relative hourly volatility:

\begin{eqnarray*}
\tilde{Q}(-K,0) &=&1.02705\times 10^{11}, \\
\sigma _{Q}^{rel}(-K) &=&0.01976.
\end{eqnarray*}

For the relative net demand curve density the values are reported in
the following table
\begin{table}[h!]\label{tbl:2}
\begin{center}
\begin{tabular}{|c|c|c|c|}
\hline
$k$ & $q(k,0)$ [in 10$^{11}$] & $\sigma _{q}^{rel}(k)$ -- hourly & $a(k)$ \\ 
\hline\hline
-6 & 0.95314 & 0.04883 & 0.11903 \\ \hline
-5 & 1.41994 & 0.04655 & 0.29142 \\ \hline
-4 & 2.35893 & 0.01706 & 0.25250 \\ \hline
-3 & 0.82541 & 0.04423 & 0.36708 \\ \hline
-2 & 0.14050 & 0.04877 & 0.36752 \\ \hline
-1 & 4.13487 & 0.03744 & 0.29380 \\ \hline
0 & 0.21397 & 0.00461 & 0.21991 \\ \hline
1 & 9.95599 & 0.00379 & 0.25219 \\ \hline
2 & 4.61052 & 0.03496 & 0.36316 \\ \hline
3 & 3.51037 & 0.00653 & 0.15248 \\ \hline
4 & 2.42507 & 0.01224 & 0.19830 \\ \hline
5 & 0.14219 & 0.00036 & 0.34405 \\ \hline
6 & 2.70257 & 0.00969 & 0.13387 \\ \hline
\end{tabular}
\end{center}
\caption{Statistics for the relative net demand curve density}
\end{table}

We then simulated the lognormal model (\ref{logbigQ}), (\ref{logq})\ using $%
N=100$ scenarios and approximated the value of a call option:

\begin{equation*}
C(\psi)=\frac{1}{N}\sum_{\omega =1}^{N}\max (\pi (T,\omega )-\psi,0),
\end{equation*}
where $\psi$ is a strike price for an expiration $T$=0.02, that is,
roughly equal to one week. We then calculated the implied volatility
of the call for each strike price, that is, the value $\sigma
^{imp}($Strike$)$ such that the Black-Scholes value of the call option
equals $C(K)$, and this for each $K$. The resulting function $\sigma
^{imp}($Strike$)$ as a function of the strike price ( the smile curve)
is reported in Figure \ref{fig:3} below.

\begin{figure}[!h]
\begin{center}
  \includegraphics[width=0.5\textwidth]{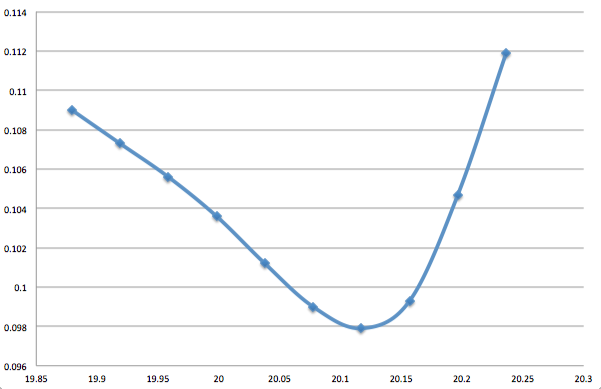}
  \caption{Implied volatility as a function of the strike price. Call
    option with 1 week expiration, $\pi(0)=20.16$, and zero interest rate.}
  \label{fig:3}
\end{center}
\end{figure}

As expected, the smile is fairly pronounced, which is an indicator that our
model engenders fat tails in the risk-neutral distribution of $\pi(T)$.


\section{Appendix} 

\begin{lemma}\label{lem:2}
Let
\begin{eqnarray*}
  \mu _{P}(x,t) &=&-\frac{\mu _{Q_{A}}(P_A(x,t),t)+\frac{1}{2}\frac{\partial
      ^{2}Q_{A}}{\partial p^{2}}(P_A(x,t),t)\sigma _{P}^{2}(x,t)+\frac{\partial
      \sigma _{Q_{A}}}{\partial p}(P_A(x,t),t)\sigma _{p}(t)}{\frac{\partial Q}{%
      \partial p}(P_A(x,t),t)} \\
  \sigma _{P}(x,t) &=&\frac{\sigma _{Q_{A}}(P_A(x,t),t)}{\frac{\partial Q_{A}}{%
      \partial p}(P_A(x,t),t)} \\
  b_{P}(x,s,t) &=&b_{Q_{A}}(P_A(x,t),t),s,t)
\end{eqnarray*}

Then $P_A$ is a semimartingale and satisfies
\begin{equation}
  dP_A(x,t)=\mu _{P}(x,t)dt+\sigma_{P}(x,t)\int\limits_{0}^{S}b_{P}(x,s,t)W(ds,dt)  \label{equforP}
\end{equation}%
\end{lemma}

\begin{proof}
By definition:
\begin{equation}
  Q_{A}(P_A(x,t),t)=x  \label{ItoWE}
\end{equation}

We suppose that (\ref{equforP})\ holds and apply the Ito-Wentzell formula
(see e.g., \cite{Krylov:09})\ to both sides of (\ref{ItoWE})\ yields:
\begin{align*}
  &\mu _{Q_{A}}(P_A(x,t),t)+\frac{\partial Q_{A}}{\partial
    p}(P_A(x,t),t)\mu
  _{P}(x,t)\\
  &\qquad\qquad\qquad\ +\frac{1}{2}\frac{\partial ^{2}Q_{A}}{\partial
    p^{2}}%
  (P_A(x,t),t)\sigma _{P}^{2}(x,t)+\frac{\partial \sigma
    _{Q_{A}}}{\partial p}%
  (P_A(x,t),t)\sigma _{P}(x,t) =0, \\
  &\sigma _{Q_{A}}(P_A(x,t),t)b_{Q_{A}}(P_A(x,t),s,t)+\frac{\partial
    Q_{A}}{%
\partial p}(P_A(x,t),t)\sigma _{P}(x,t)b_{P}(x,s,t) =0
\end{align*}%
\end{proof}

{\it Proof of Theorem \ref{thm:2}}

For fact (F1), we first check assumptions (RF)\ in \cite{KallRhein:09}:

(i) Since $L_L(\vartheta ,t)$ is a semimartingale, it is a strong integrator.

(ii) We calculate 
\begin{equation*}
  \frac{\partial ^{2}L_L(\vartheta ,t)}{\partial \vartheta ^{2}}=\frac{\partial 
  }{\partial x}P_{A}(\vartheta ,t)=\frac{1}{\frac{\partial Q^{A}}{\partial p}%
    (P_{A}(\vartheta ,t))}.
\end{equation*}

Since $Q^{A}$ is strictly decreasing and twice differentiable in $p$, then
its inverse $P_{A}$ is differentiable in $\vartheta $ and $\frac{\partial
^{2}L_L(\vartheta ,t)}{\partial \vartheta ^{2}}$ is continuous in $\vartheta $.

(iii)\ Since $\frac{\partial L_L(\vartheta ,t)}{\partial \vartheta}%
=P_A(\vartheta ,t)$ is a semimartingale, it is a strong integrator.

(iv) By the Lemma \ref{lem:2} the quadratic variation of $L_L(\vartheta ,.)$ is equal to
\begin{equation*}
  \lbrack L_L(\vartheta ,.), L_L(\vartheta ,.)]_{t}=\int_{0}^{t}\int_{0}^{\vartheta
  }\left( \frac{\sigma _{Q}(P_A(x,t),t)}{\frac{\partial Q}{\partial p}(P_A(x,t),t)}%
  \right) ^{2}dt.
\end{equation*}

This is clearly strictly increasing if $\sigma _{Q}(P_A(x,t),t)$ is uniformly
bounded away from zero, which holds if $\sigma _{Q}(p,t)$ is uniformly
bounded away from zero on $0<p<S$.

Thus assumptions (RF) are satisfied. We now define the market price
of risk as
\begin{equation*}
\lambda ^{(x)}(s,t)=\frac{\mu _{P}(x,t)}{\sigma _{P}(x,t)}.
\end{equation*}

Since $\mu _{P}(x,t)$ is bounded and $\sigma _{P}(x,t)$ is non-zero,
then the assumption (UB) in \cite{KallRhein:09} is satisfied. To check
the assumption (UI) we refer the reader to \cite{KallRhein:09} for the
definition of $\lambda _{t}^{(n)}$ and $[M^{(n)},M^{(n)}]$. Since both
$Q_{A}$ and $Q_{L}$ are uniformly bounded, the integrand
$\lambda_{t}^{(n)}$ is bounded. By the same reason, and because of the
continuity (in time) of $Q_{A}$, the quadratic variation
$[M^{(n)},M^{(n)}]$ is bounded. Thus (UI) holds, and Theorem 3.5 in
\cite{KallRhein:09} holds, proving Fact (F2).

Fact (F2) follows directly from the comments after Theorem 3.5 in
\cite{KallRhein:09}, which we quote here, while adjusting for our
notation and formulae numbering:

``Consider now again the dynamics \eqref{eq:bb04} of the real wealth
process, and let $\theta$ be such that $\int L_L(\theta, ds)$ is
bounded from below (the transaction costs term in \eqref{eq:bb04} can
be avoided by the large trader by using only tame strategies). An {\it
  arbitrage opportunity} is an admissible strategy such that we have
for the associated real wealth process $V^{\theta}$ that
$V^{\theta}(0)\le 0,\ V^{\theta}(T)\ge 0$ $\mathbb{P}$-a.s., and
$\mathbb{P} (V^{\theta}(T) > 0) > 0$. By Theorem 3.5 in
\cite{KallRhein:09}, there exists a probability measure $\mathbb{Q}$
such that $\int L_L(\theta, ds)$ is a $\mathbb{Q}$-local martingale,
hence a supermartingale. It follows now from the dynamics
\eqref{eq:bb04} of the real wealth process that
$\mathbb{E}_{\mathbb{Q}}[V^{\theta}(T)]\le V^{\theta}(0)$ which, as
$\mathbb{Q}$ is equivalent to $\mathbb{P}$, excludes arbitrage
opportunities for the large trader.``

We can now prove the continuity of $\pi (t)$. Since by assumption
uncross orders result in $Q_{L}$ being continuous in time, a
discontinuity can arise only from cross orders. Suppose that the large
trader is a net buyer (for a sell strategy the argument is identical),
i.e, that there is an $\varepsilon >0$ such that for all $0\leq \delta \leq
\varepsilon $
\begin{equation*}
  \theta (t+\delta )-\theta (t)\geq 0.
\end{equation*}



Suppose that
\begin{align}
  \label{eq:1}
  \pi(t)-\pi(t_-)\ge\epsilon_2 >0.
\end{align}
Then, since $\partial Q_A/\partial p >0$ and $Q_A$ is assumed to be
twice differentiable
$$-Q_A(\pi(t),t)+Q_A(\pi(t_-),t)\ge\epsilon_3>0.$$
By continuity in time of $Q_A$
$$-Q_A(\pi(t),t_+)+Q_A(\pi(t_-),t)\ge\epsilon_4>0.$$
An examination of terms in Lemma 3.3 in \cite{BankBaum:04} shows that
discontinuous strategies are suboptimal for the trader, i.e.
\begin{align}
  \label{eq:2}
  \theta(t_+)-\theta(t)=0.
\end{align}
By definition
$$\theta(t_+)-\theta(t) = -Q_A(\pi(t),t_+)+Q_A(\pi(t_-),t).$$
However, this shows, that \eqref{eq:1} contradicts \eqref{eq:2}. Thus
$\pi$ is continuous, proving fact (F4), and 
\begin{align}
  \label{eq:3}
  Q_L(\pi(t), t)-Q_L(\pi(t_-), t)=0.
\end{align}
By definition also
\begin{align}
  \label{eq:4}
  \theta(t_+)-\theta(t) &= Q_L(\pi(t),t_+)-Q_L(\pi(t_-),t)\\
  &= Q_L(\pi(t),t_+)-Q_L(\pi(t),t)+Q_L(\pi(t),t) -Q_L(\pi(t_-),t).\nonumber
\end{align}
Combining \eqref{eq:2}, \eqref{eq:3} and \eqref{eq:4} we obtain that
$$Q_L(\pi(t),t_+)-Q_L(\pi(t),t)=0.$$
This shows that the net demand curve is continuous in $t$, thus
proving fact (F3).
 




Finally, we build the curve:
\begin{equation*}
  Q=Q_{L}+Q_{A}.
\end{equation*}

Suppose that a new atomistic trader comes to the market. She will
trade at a price $\pi (t)$ and will have liquidity costs:
\begin{equation*}
  \int\limits_{0}^{t}L(\theta (u),du)=\int_{0}^{t}\theta (u)d\pi (u).
\end{equation*}

Since $\int\limits_{0}^{t}L(\theta (u),du)$ is a $\mathbb{Q}$-martingale,
therefore $\pi $ is a $\mathbb{Q}$-martingale, too.

\begin{flushright}
  $\qedsymbol$
\end{flushright}

\bibliographystyle{alpha}
\bibliography{../AUX/finance}



\end{document}